\documentclass{elsarticle}

\usepackage{lineno,hyperref}
\usepackage{amsmath,graphicx,bm,amssymb,amsthm,enumerate,epsfig,psfrag}

\modulolinenumbers[5]

\newcommand{\beq}{\begin{equation}}
\newcommand{\eeq}{\end{equation}}
\newcommand{\bmat}{\begin{pmatrix}}
\newcommand{\emat}{\end{pmatrix}}
\newcommand{\pdim}{p}
\newcommand{\kdim}{K}
\newcommand{\qdim}{Q}
\newcommand{\lam}{\lambda}

\newcommand{\A}{\bo A}
\newcommand{\Nabla}{\bm {\nabla}}

\newcommand{\B}{\mathbf{B}}
\newcommand{\I}{\bo I}

\renewcommand{\(}{\left(}
\renewcommand{\)}{\right)}

\newcommand{\ndim}{n}
\newcommand{\hop}{\mathrm{T}}
\newcommand{\tr}{\mathrm{Tr}}     

\newcommand{\bom}[1]{\boldsymbol{#1}}    
\newcommand{\bo}[1]{\mathbf{#1}}              
  
\newcommand{\be}{\beta}
\newcommand{\mV}{\mathcal V} 
\newcommand{\PDH}{\mathcal{S}}   
\newcommand{\Tr}{\mathrm{Tr}}
\newcommand{\paino}{\sl} 
\newcommand{\M}{\bom \Sigma}
\newcommand{\R}{\mathbb{R}}

\newcommand{\Z}{\bo X}
\newcommand{\z}{\bo x}
\newcommand{\Q}{\M}
\newcommand{\D}{\bo D}
\newcommand{\x}{\bo x}
\newcommand{\dI}{d_{{\rm KL}}}
\newcommand{\dE}{d_{{\rm E}}}

\newcommand{\dR}{d_{{\rm R}}}

\newtheorem{condition}{Condition}

\theoremstyle{remark}
\newtheorem{remark}{Remark}

\newtheorem{lemma}{Lemma}

\newtheorem{theorem}{Theorem}
\newtheorem{proposition}{Proposition}

\newcommand{\norm}[1]{\left\lVert#1\right\rVert}
\renewcommand{\S}{d'} 
\newcommand{\C}{\mathbf{C}}

\newcommand{\U}{\mathbf{U}}

\newcommand{\bmu}{\bom \mu}

\journal{Journal of \LaTeX\ Templates}

\begin{document}

\begin{frontmatter}


\title{Simultaneous penalized M-estimation
of covariance matrices using geodesically convex optimization}



\author[address1]{Esa Ollila\corref{mycorrespondingauthor}}
\cortext[mycorrespondingauthor]{Corresponding author}
\ead{esa.ollila@aalto.fi}
\address[address1]{Aalto University, Finland}
\author[address2]{ Ilya Soloveychik}
\address[address2]{The Hebrew University of Jerusalem, Israel }
\author[address3]{David E. Tyler\fnref{myfootnote1}}
\address[address3]{Rutgers -- The State University of New Jersey, USA}
 \fntext[myfootnote1]{Research partially supported by the National Science Foundation Grant No.\  DMS-1407751. Any Opinions, findings and conclusions or recommendations expressed in this material are those of the author(s) and do not necessarily reflect those of the National Science Foundation.}
 \author[address2]{Ami Wiesel\fnref{myfootnote2}}
 \fntext[myfootnote2]{The author gratefully acknowledge Israel Science Foundation grant 1339/15}

\begin{abstract}
A common assumption when sampling $p$-dimensional observations from $K$ distinct group is the equality of the covariance matrices.  In this paper, we propose two penalized $M$-estimation approaches for the estimation of the covariance or scatter matrices under the broader assumption that they may simply be close to each other, and hence roughly deviate from some positive definite ``center''.  The first approach begins by generating a pooled $M$-estimator of scatter based on all the data, followed by a penalised $M$-estimator of scatter for each group, with the penalty term chosen so that the individual scatter matrices are shrunk towards the pooled scatter matrix. In the second approach, we minimize the sum of the individual group $M$-estimation cost functions together with an additive joint penalty term which enforces some similarity between the individual scatter estimators, i.e. shrinkage towards a mutual center.  In both approaches, we utilize the concept of geodesic convexity to prove the existence and uniqueness of the penalized solution under general conditions.  We consider three specific penalty functions based on the Euclidean, the Riemannian, and the Kullback-Leibler distances. In the second approach, the distance based penalties are shown to lead to estimators of the mutual center that are related to the arithmetic, the Riemannian and the harmonic means of positive definite matrices, respectively. A penalty based on an ellipticity measure is also considered which is particularly useful for shape matrix estimators. Fixed point equations are derived for each penalty function and the benefits of the estimators are illustrated in regularized discriminant analysis problem. 
\end{abstract}


\end{frontmatter}


\section{Introduction}
\label{sec:intro}

Many multivariate statistical applications require the simultaneous estimation of the covariance matrices $\bm\Sigma_1,\dots,\bm\Sigma_K$ of a set of multivariate measurements 
on $K$ distinct groups.  Often the sample sizes $n_k, k = 1, \dots, K$, of each group are small relative to the dimension $p$, which makes estimating the individual 
covariance matrices a challenge. Quite often, though, based on the physical properties of the underlying measured phenomena or experience with similar datasets, one may
postulate the existence of common features or similarities among the estimated covariance matrices. This prior knowledge can be incorporated into the estimation problem by 
either modeling the covariance matrices as having some common structure or by pooling the data from the $K$ groups.

In this paper we focus on data pooling techniques via regularization. The use of pooling and regularization methods assume the distinct covariance matrices share some common features, without necessarily modeling the common features.  A prominent example of this approach is Friedman's regularized discriminant analysis \citep{friedman1989regularized}.
A similar approach to estimating precision matrices, i.e.\ inverse covariance matrices, was treated in \cite{lee2015joint}.
The goal of Friedman's regularized discriminant analysis approach is to strike a balance between quadratic and linear discriminant analysis (QDA/LDA) in the under-sampled 
scenario via shrinkage regularization. In \cite{friedman1989regularized} it was illustrated that it is often beneficial to shrink the class Sample Covariance Matrices (SCM) 
towards the pooled SCM.

The methods proposed in these works were developed under the assumption of sampling from multivariate normal distributions. Consequently, they
tend to depend on variants of the SCM estimator and are not resistant to outliers nor robust against heavier tailed distributions. From this perspective, taking into 
account the non-Gaussianity of measurements in many real world applications, the statistical community has become increasingly aware of the advantage of more robust and 
resistant multivariate methods. This, in particular, led to development of the family of the $M$-estimators of multivariate scatter \citep{huber1981book, maronna1976robust, 
tyler1987distribution}, as well as families of high-breakdown point scatter estimators such as the MVE and MCD estimator \citep{rousseeuw1985multivariate}, 
the $S$-estimators \citep{davies1987asymptotic}, and the $MM$-estimators \citep{tatsuoka2000uniqueness}, among others.  There appears, though, to be little work on robustness in the context of 
joint covariance estimation and its application to regularized discriminant analysis and other problems. The intent of this paper is to address this issue. 
We focus on $M$-estimation methods, which unlike the high breakdown point methods, are readily amenable to the sparse data setting and regularization.

Our aim is to propose robust versions of the SCM based shrinkage covariance estimators proposed in \cite{friedman1989regularized} for regularized 
discriminant analysis (RDA) in the sparse data setting. The approach used in \cite{friedman1989regularized} is based on taking a convex combination
of the individual SCM and the pooled SCM. Such an approach, though, does not directly generalize when using $M$-estimators of scatter, since
the $M$-estimators are not defined when the data within a group is sparse. 
Rather, in our approach we apply penalization to $M$-estimation loss functions. When using such loss functions which correspond to bounded influence $M$-estimators of scatter,
though, one encounters a non-convex optimization problem in Euclidean space. Here, the concept of geodesic convexity ($g$-convexity) plays a crucial role, which basically means switching to a different metric over the set of positive definite matrices, for which the loss function is then convex in this metric. The use of $g$-convexity in covariance estimation
was introduced in \cite{wiesel2012geodesic} and has subsequently been utilized in related works, e.g. \cite{wiesel2012unified,zhang2013multivariate,sra2015conic,ollila2014regularized}. See \cite{wiesel2015structured} for a nice overview of usage of $g$-convexity in covariance matrix estimation problems. 
Introducing additive $g$-convex penalty terms to the loss functions, keeps the optimization problem $g$-convex.

Two penalized $M$-estimation approaches are introduced for the problem of joint estimation of group covariance matrices. The first approach begins by defining 
a pooled $M$-estimator of scatter based on all the data, followed by a penalized $M$-estimator of scatter for each group, with the penalty term chosen so that 
the individual scatter matrices are shrunk towards the pooled scatter matrix.  In the second approach, we minimize the sum of the individual group $M$-estimation 
loss functions together with an additive joint penalty term which enforces some similarity between the individual scatter estimators, i.e.\ shrinkage towards a 
mutual center. Hence, in the second approach, the individual covariance matrices and their mutual center are estimated simultaneously. In both approaches, we 
consider three $g$-convex penalty functions based on the Euclidean, the Riemannian, and the information theoretic 
(Kullback-Leibler) distances. In the second approach, these penalties are shown to lead to estimators of the mutual center 
that are related to the arithmetic, the Riemannian and the harmonic means of positive definite matrices, respectively. We also consider a penalty based on
an ellipticity measure for positive definite matrices, which shrinks the individual estimators towards a common shape matrix rather than a common scatter
matrix. 

The rest of the paper is organized as follows.  Section~\ref{el_pop_sec} introduces our penalized  $M$-estimation approaches for estimating the unknown $\kdim$ 
scatter matrices $\{\Q_k\}_{k=1}^\kdim$ and their joint center $\Q$. Examples of  $g$-convex loss functions, including the Gaussian, Huber's and Tyler's loss functions, are given.  Section~\ref{sec:gconv} provides a brief introduction  to $g$-convex functions of positive definite symmetric (PDS) matrices. In Section~\ref{sec:pencost} 
examples of $g$-convex penalty/distance functions are given. In addition, we show that the KL-distance and the ellipticity distance are $g$-convex, and when used for
defining a center for a given $\{\M_k\}_{k=1}^\kdim$ yield weighted harmonic means of positive definite matrices. 
In Section~\ref{sec:uniq} we derive general conditions for uniqueness of the solution as well as derive fixed point algorithms for their computation. Section~\ref{sec:tyl} considers existence and uniqueness conditions separately for Tyler's loss function. Section~\ref{sec:CV} describes a cross validation procedure for penalty parameter selection. In Section~\ref{sec:RDA} we illustrate the application of the proposed scatter matrix estimators to
regularized discriminant analysis   
and illustrate the performance of  RDA rules via a small simulation study and a data example. 
Section~\ref{sec:concl} concludes the paper.  Proofs are given in the Appendix.

{\it Notation:} Let $\mathcal{S}(p)$ be the open cone of positive definite $p \times p$ symmetric matrices, and let $\I$ be the identity matrix of proper dimension. 
On $\mathcal{S}(p)$, we denote the Frobenius norm by $\norm{\cdot}_{\rm F}$, the spectral norm by $\norm{\cdot}_2$, and the determinant by $|\cdot|$.


\section{Problem Formulation} \label{el_pop_sec}

\subsection{General Setting} \label{gen_set}

The multivariate $M$-estimators were introduced in \cite{maronna1976robust} as generalizations of the maximum likelihood estimators for an elliptically symmetric
multivariate distribution. An absolutely continuous random vector $\z \in \R^\pdim$ is said to have a real elliptically symmetric (RES) distribution with center of
symmetry $\bom \mu$ and scatter matrix parameter $\Q \in \PDH(\pdim)$, if it has a density of the form
\begin{equation}
f(\x | \Q) =   C_{\pdim,g} |\Q|^{-1/2} g\{(\x - \bom \mu)^\top\Q^{-1}(\x-\bom \mu)\},  
\end{equation}
where $C_{\pdim,g}$ denotes the normalizing constant, and $g: \R^+ \to \R^+$ is viewed as a density generator. Here, $\R^+ = \{ x \in \R | x \ge 0 \}$. For simplicity, we state
$\x \sim \mathcal{E}_\pdim(\bom \mu,\Q,g)$.  The function $g$ determines the radial distribution of the elliptical population and hence the degree of its ``heavy-tailedness''.  
The scatter matrix $\Q$ is proportional to the covariance matrix whenever the second moments exist, and serves as a
generalization of the covariance matrix when the second moments do not exist. There is a extensive literature on the properties of elliptical distributions. The
elliptical family includes many widely used multivariate distributions such as as Gaussian, compound Gaussian, $K$-distributions, among many others. 
For a thorough treatment of elliptical distributions and their generalizations see e.g., \cite{frahm2004generalized,ollila2012complex}. 

Consider samples from $K$ distinct groups of $p$-dimensional measurements, 
\begin{equation}
\x_{11},\dots,\x_{1n_1},\quad\dots,\quad\x_{K1},\dots,\x_{\kdim n_K},
\end{equation}
with group $\Z_k = \{\x_{k 1},\dots,\x_{k n_k}\}$ 
have sample size $n_k$, $k = 1, \ldots, K$. Let
 \beq \label{eq:pi_k} 
 N=\sum_{i=1}^\kdim n_k \quad \mbox{and} \quad \pi_k = \frac{n_k}{N}, ~ \mbox{for} ~ k = 1, \ldots, K    
\eeq
denote the total sample size and the relative sample sizes of each of the $K$ groups, respectively.  The measurements are assumed to be  mutually independent and within each group they are assumed to be identically distributed. 

In our development, we first presume the measurements within the different groups follow elliptical distributions with known centers of symmetry, which we take without loss
of generality to be $\bom \mu_k = \bo 0$ for $k = 1, \ldots, K$. The assumption of having known centers is to be discussed later. Hence, we assume the random sample of
the measurements for the $k$th group comes from an $\mathcal E_\pdim(\bo 0,\M_k,g_k)$ distribution, $k = 1, \ldots, K$, with possibly different scatter matrices $\Q_k$.
The negative log-likelihood for this scenario, ignoring the normalizing constant $C_{\pdim,g}$, is proportional to
\beq 
\label{eq:likelihood} 
\mathcal L(\Q_1, \ldots, \Q_k) =  \sum_{k=1}^K \pi_k \mathcal L_k (\Q_k),
\eeq 
where 
\begin{align} \label{eq:Lk} 
\mathcal L_k (\Q_k) 
&=  \frac{1}{n_k}\sum_{i=1}^{n_k} \rho_k( \x^\hop_{ki} \Q^{-1}_k \x_{ki} ) - \log |\Q^{-1}_k|,
\end{align} 
and $\rho_k(t) = -2 \log g_k(t)$. The nature of $M$-estimation is to then divorce the estimators obtained from minimizing \eqref{eq:likelihood}
from the distributions that generated the negative log-likelihood function. When using the sample covariance matrix, for example, one does need to assume 
it is based on a sample from a multivariate normal distribution or even from an elliptical distribution. In general, for respective loss functions 
$\rho_k: \R^+ \to \R^+$, not necessarily related to any $g_k$, a minimizer \eqref{eq:likelihood} represents an $M$-estimator of scatter. For more detail discussions 
on the concepts underlying $M$-estimation and other robust methods, see \cite{huber1981book,hampel1986book,maronna2006book}. 

Minimizing \eqref{eq:likelihood} over $\Q_1, \ldots, \Q_K \in \PDH(\pdim)$ is equivalent to minimizing \eqref{eq:Lk} individually over $\Q_k \in \PDH(\pdim)$
for $k = 1, \dots, K$, i.e.\ obtaining the individual $M$-estimators of scatter for each group. One drawback to this approach is that the individual $M$-estimators of
scatter do not exist when $n_k < \pdim$ \citep{kent1991redescending}, and do not differ substantially from the sample covariance matrix when $n_k$ is only slightly larger than $\pdim$.
Consequently, for sparse group data, we need to pool the information in the different groups and hence presume that the scatter matrices are somewhat similar across the
groups. The most extreme and most common assumption is that the scatter matrices are equal across groups. Here, though, we make no strong model assumptions regarding
the different scatter matrices, but rather propose the following two penalization approaches.

\emph{Proposal 1: Regularization towards a pooled scatter matrix.}  A pooled $M$-estimator of scatter, obtained by pooling together the data from each of the $K$ groups, can be defined as a minimum of 
\beq 
\label{eq:Mpooled} 
\mathcal L(\Q) =  \sum_{k=1}^K \pi_k \mathcal L_k (\Q) = \frac{1}{N} \left\{\sum_{k=1}^K \sum_{i=1}^{n_k} \rho_k( \x^\hop_{ki} \Q^{-1} \x_{ki} ) \right\} - \log |\Q^{-1}|.
\eeq 
over $\Q \in \PDH(\pdim)$.  Penalized $M$-estimators of scatter for the individual groups can then be defined as a solution to the optimization problem
\beq 
\label{eq:kpenfun} 
\min_{ \Q_k \in \PDH(\pdim) } \left\{\mathcal L_k (\Q_k)+ \lambda \, d(\Q_k,\hat \Q)  \right\}, \quad k = 1, \ldots K,
\eeq 
where $\hat \Q$ is minimizer of \eqref{eq:Mpooled}, $d(\Q_k,\Q)$ represents a penalty based on distances between
$\Q_k$ and $\Q$, and $\lambda$ is positive tuning parameter, chosen by the user, which balance the interplay 
between unrestricted $M$-estimation of scatter and shrinkage towards $\Q$. 
Equivalently, we can write optimization program in \eqref{eq:kpenfun} in the form 
 \beq 
 \label{eq:kpenfun2} 
 \min_{ \Q \in \PDH(\pdim) } \left\{  \be \mathcal  L_k (\Q_k)  + (1-\be)  d(\Q_k,\hat \Q)    \right\}  , 
 \eeq 
where  penalty parameter $\be \in (0,1]$ is one-to-one with $\lam>0$ via  mapping  $\lam = (1-\be)/\be$.   Formulation \eqref{eq:kpenfun}  
is in many ways more  instructive as it depicts the role of the penalty term in more lucid manner:  one may view the penalty parameter $\be $ 
as  a "probability" or degree of belief  one assigns on the cost function $\mathcal L_k (\Q_k)$ relative to the penalty term $d(\Q_k,\hat \Q)$. 
Moreover, $\be$ is conveniently on scale $(0,1]$.  
The latter formulation  \eqref{eq:kpenfun2} via regularization parameter $\be$ will be used when constucting the fixed point algorithms in Section~\ref{sec:uniq}.   
Examples of penalty functions $d(\Q_k,\Q)$  and their properties are addressed in Section~\ref{sec:pencost}.  

\emph{Proposal 2: Joint regularization enforcing similarity among the group scatter matrices.} 
Rather than first defining a pooled scatter matrix, our second proposal simultaneously estimates the group scatter matrices $\Q_k$ along
with their `center' $\Q$.  The optimization program is now
\beq 
\label{eq:penfun} 
\underset{\{ \Q_k \}_{k=1}^K,\Q \in \PDH(\pdim)^\pdim}{\mathrm{minimize}} \, \sum_{k=1}^K \pi_k \left\{ \mathcal L_k (\Q_k)+\lambda \, d(\Q_k,\Q) \right\}.
\eeq 
The penalty term $d(\Q_k,\Q)$ is as before, but now is viewed as enforcing similarity among the $\Q_k$-s, and the `center' $\Q$ is now viewed as
an `average' of the $\Q_k$-s. Note again that it is possible to write \eqref{eq:penfun}  via penalty parameter $\be \in (0,1]$ (where $\beta=1/(1+\lambda)$)  as 
in \eqref{eq:kpenfun2}  in which case the term $ \mathcal L_k (\Q_k)+\lambda \, d(\Q_k,\Q)$ in \eqref{eq:penfun} 
is replaced by $\be \mathcal L_k (\Q_k)+ (1-\be) \, d(\Q_k,\Q)$. 
Note that for fixed $\Q_1, \ldots, \Q_K$, the value of $\Q$ is given by
\beq\label{mean_def}
\Q(\bom \pi) 
= \underset{\M \in \PDH(\pdim)}{\arg \min} \sum_{i=1}^\kdim \pi_k \,  d(\Q_k,\Q),
\eeq 
which represents the {\paino weighted mean} associated with the distance $d$. For example, the Euclidean, or Frobenius, distance $d_{{\rm F}}(\Q_k,\Q) = \left\{\tr [(\Q_k - \Q)^2]\right\}^{1/2}$ gives the standard weighted 
{\paino arithmetic mean} $\Q_{{\rm F}}(\bom \pi) = \sum_{k=1}^\kdim \pi_k \Q_k$. 

Modest modifications to Proposals 1 and 2 can be considered. For example, one might consider replacing the tuning constant $\lambda$ in either
proposal with individual tuning constants, say $\lambda_k, \, k = 1, \ldots, K$. Typically one tends to choose a larger tuning constant when sample sizes
are smaller. However, in our proposals, this does not seem to be necessary since for a particular group, say group $j$, for which $n_j$ is the smallest,
the term $d(\Q_j,\Q)$, in either proposal, affects $\Q_j$ more then the other groups since group $j$ affects the value of $\Q$ the least.
Another modification to proposal 1 is to consider other pooled estimates of scatter. In particular, if the total sample size $N$ is small,
and in particular if $N < p$, then we recommend adding a penalty term to \eqref{eq:Mpooled} itself, say one which penalized $\Q$ for deviations 
from $\bo I$ or deviations from proportionality to $\bo I$, see e.g., \cite{ollila2014regularized} or \cite{wiesel2015structured}. We also recommend such an additional penalty term 
to \eqref{eq:penfun} in Proposal~2 when $N$ is relatively small.

For the special case, $\rho_k(t) = t$ for $k = 1, \ldots K$, the solution for $\Q_k$ in Proposal~1 is 
\begin{equation} \label{eq:Friedman}  
\bo S_k(\beta) = \beta \bo S_k + (1 - \beta) \bo S, 
\end{equation}
where $\bo S_k = \frac{1}{n_k} \sum_{i=1}^{n_k} \z_{ki} \z_{ki}^\top$ is the sample covariance matrix for the $k$th group, 
$\bo S = \sum_{k=1}^K \pi_k \bo S_k$ is the pooled sample covariance matrix, and $\beta = 1/(1+\lambda)$. Note that as the
tuning constant $\lambda \rightarrow \infty$, $\bo S_k(\beta) \rightarrow \bo S$, and as $\lambda \rightarrow 0$,
$\bo S_k(\beta) \rightarrow \bo S_k$. The estimator \eqref{eq:Friedman} is the one proposed by Friedman in  \cite{friedman1989regularized}
in the context of regularized discriminant analysis. Hence, Proposal~1 can be view as a direct generalization of Friedman's estimator.

\subsection{Examples of loss functions} \label{sec:loss} 

Throughout, we assume the loss functions $\rho_k(t)$, $k = 1, \ldots K$, satisfy the following condition:
\begin{condition} \label{cond:rho}
The loss functions $\rho_k(t)$, $k=1,\ldots,\kdim$ are nondecreasing and continuous for $0 < t < \infty$. In addition, $\rho_k(t)$ is convex in $\log t$, i.e.\
the function $r_k(x) = \rho_k(e^x)$ is convex for $-\infty < x < \infty$.
\end{condition}
Typically, the loss functions $\rho_k(t)$ will be the same for $k = 1, \ldots, K$, but our general development allows for the case when they may differ.
Also, the loss functions are often standardized so that the estimators obtained by minimizing \eqref{eq:Lk} are Fisher consistent
when the $k$th sample represents a random sample from the Gaussian distribution $\mathcal N_\pdim(\bo 0, \M_k)$.  This holds if and only if 
$E[\psi_k(\chi^2_p)] = p$, where $\psi_k(t) = t u_k(t)$ and $u_k(t) = \rho_k'(t)$,

Below we provide some common examples of loss functions $\rho_k$ often encountered in the literature and used in multivariate analysis, along
with their corresponding weight functions $u_k$. The weight functions themselves are needed in section \ref{sec:uniq} to represent the corresponding
$M$-estimating equations and in deriving fixed-point algorithms for the estimators.

(i) {\bf{Gaussian loss function}}.  
The density generator for $\mathcal N_\pdim(\bo 0, \M)$ is $g(t) = \exp(-t/2)$.   Hence, the corresponding  
loss and weight functions are $\rho_{\rm G}(t)=t$ and $u_G(t)=1$ respectively.  
The corresponding objective function for the $k$th sample, i.e.\ \eqref{eq:Lk}, is then
\beq \label{eq:Lg} 
\mathcal L_{{\rm G},k} (\Q_k)  = \Tr( \M^{-1}_k \bo S_k) - \log|\Q^{-1}_k|
\eeq 
where $\bo S_k$ again denotes the sample covariance matrix of the $k$-th sample. 


(ii) {\bf{$t$ loss functions}}:
The density generator for a $\pdim$-variate elliptical $t$-distribution on $\nu > 0$ degrees of freedom is $g_\nu(t) = (\nu + t)^{-\frac{1}{2}(\nu + \pdim)}$.   
Hence, the corresponding loss and weight functions are $\rho_{\nu}(t) = (\nu + p)\log(\nu + t)$ and $u_\nu(t) = (\nu + p)/(\nu + t)$ respectively.  The resulting
$M$-estimators of scatter are not Fisher consistent at a multivariate Gaussian distribution. However, one can obtain such a Fisher consistent version of the
$t$ $M$-estimators by taking the loss function to be $\rho_{\nu,b}(t) \equiv \rho_{\nu}(t)/b$, with $b$ chosen so that $b = E[\psi_\nu(\chi^2_p)]/p$ and where
$\psi_\nu (t) = t u_\nu (t)$. This gives $b = \{(\nu+p)/p\} E[\chi^2_p/(\nu +\chi^2_p)]$.

(iii) {\bf{Huber's loss function}}: In his seminal work, Huber \cite{huber1964robust} proposed a family of univariate heavy-tailed distributions 
often referred to as ``least favourable distributions'' (LFDs). A LFD corresponds to a symmetric unimodal distribution which 
follows a Gaussian distribution in the middle, and a double exponential distribution in the tails. The corresponding maximum likelihood estimators 
are then referred to as Huber's $M$-estimators.  The extension of Huber's $M$-estimators to the multivariate setting, is usually defined as a generalization
of the corresponding univariate $M$-estimating equations to the multivariate setting, see e.g., \cite{maronna1976robust}. 

Here, we illustrate how Huber's $M$-estimators of multivariate scatter can be viewed as maximum likelihood estimators for a family of heavy-tailed
$\pdim$-variate elliptical distributions, namely those with density generator of the form 
$g_{\rm H}(t;c) = \exp\{- (1/2) \rho_{\rm H}(t;c)\}$,where
\beq \label{eq:huber_rho} 
\rho_{\rm H}(t;c) = \begin{cases} t/b &   \ \mbox{for} \ t \leqslant c^2, \\ 
(c^2/b) \big( \log (t/c^2) + 1 \big)  & \ \mbox{for} \ t > c^2. \end{cases}
\eeq 
These distributions follow a multivariate Gaussian distribution in the middle, but have tails that die down at an inverse polynomial rate. The distribution is
a valid distribution for $c > 0$, and for the corresponding maximum likelihood estimator of scatter, i.e. the Huber $M$-estimator of multivariate scatter, the
index $c$ represents a user defined tuning constant that determines the robustness and efficiency of the estimator. The constant $b>0$ represents a scaling factor 
since it has the effect that if $\widehat{\Q}$ represents the resulting Huber's $M$-estimator of scatter whenever $b = 1$, the Huber's $M$ estimator of scatter
when $b = b_o$ is simply $b_o\widehat{\Q}$. The scaling constant $b$ is usually chosen so that the resulting scatter estimator of scatter is Fisher consistent 
for the covariance matrix at a chosen reference $\pdim$-variate elliptical distribution, commonly the $\pdim$-variate Gaussian distribution.  Given a value
of $c$, the value of $b$ needed to obtain Fisher consistency at Gaussian distributions is $b = F_{\chi^2_{\pdim+2}}(c^2) + c^2(1-F_{\chi^2_{\pdim}}(c^2))/\pdim$.

We refer to $\rho_{\rm H}(t;c)$ as Huber's loss function, since it gives rise to Huber's weight function, namely  
\[
u_{\rm H} (t;c) = \rho_{\rm H}'(t ; c) =   \begin{cases}  1/b,  
&  \ \mbox{for} \ t \leqslant  c^2 \\ c^2/(t b),  & \ \mbox{for} \ t > c^2 \end{cases} . 
\]
Thus, an observation $\z$ with squared Mahalanobis distance (MD) $t=\z^\top \M^{-1} \z$ smaller than $c^2$ receives constant weight, while observations with large MD are heavily downweighted.  

(iv) {\bf{Tyler's loss function}}:
The Gaussian loss function can be viewed as a limiting case of either a $t$ loss function or Huber's loss function by considering $\nu \rightarrow \infty$ or
$c \rightarrow \infty$ respectively. At the other extreme, i.e. as $\nu \rightarrow 0$ or $c \rightarrow 0$, one obtains Tyler's loss function
$\rho_{\rm T}(t)=\pdim \log t$, whose corresponding weight function is $u_{{\rm T}}(t)=p/t$. To obtain this limit using Huber's loss function, first note that
the Huber's $M$-estimator is not affected by replacing $\rho_{\rm H}(t;c)$ with 
$\rho^*_{\rm H}(t;c) = \rho_{\rm H}(t;c) - h(c,b)$, with $h(c,b) = c^2\{1-\log(c^2)\}/b$ being constant in $t$. Then, since $c^2/b \rightarrow p$
as  $c \rightarrow 0$, it follows that $\rho^*_{\rm H}(t;c) \rightarrow \rho_{\rm T}(t)$. 
Using this loss function, the corresponding objective function \eqref{eq:Lk} for the $k$th sample becomes 
\beq \label{acg_ll}
\mathcal{L}_{{\rm T},k}(\Q_k)= \frac{p}{n_k}\sum_{i=1}^{n_k}  \log\(\x_{ki}^\top\Q_k^{-1}\x_{ki}\)-\log|\Q_k^{-1}|.
\eeq 
A minimizer of \eqref{acg_ll} yields Tyler's \cite{tyler1987distribution} distribution-free $M$-estimator of scatter.  
Note that \eqref{acg_ll} does not have a unique minimum, since if $\Q_k$ is a minimum then so is $b \Q_k$ for
any $b > 0$. That is, Tyler's $M$-estimator estimates the shape of $\Q_k$ only. A Fisher consistent estimator of the
covariance matrix at a Gaussian distribution can be obtained by multiplying any particular minimum $\Q_k$ by
$b_k = \mbox{Median}\{\x_{ki}^\top\Q_k^{-1}\x_{ki} ; i = 1, \ldots, n_k\}/\mbox{Median}(\chi^2_p)$. In discriminant application reported in Section~\ref{sec:RDA}, this scaling is utilized.  

It is worth noting that the objective function \eqref{acg_ll} does not correspond to the negative log-likelihood of any
family of RES distributions, since $g(t) = e^{-\rho_{\rm T}(t)/2} = t^{-p/2}$ is not a valid density generator. However, 
\eqref{acg_ll} does correspond the negative log-likelihood function for a $\pdim$-variate Angular Central Gaussian (ACG) distribution \cite{tyler1987statistical}.
The ACG distribution is defined on the unit $\pdim$-sphere $\mathcal{S}^{p-1} = \{ \theta \in \R^p; \theta^\top\theta = 1 \}$, and its p.d.f.\ relative
to the uniform distribution on $\mathcal{S}^{p-1}$ has the form 
\beq
f(\x|\Q) =  C_\pdim |\Q|^{-1/2} (\x^\top  \Q^{-1}\x)^{-\pdim/2}, \quad \x \in \mathcal{S}^{p-1}.
\label{tyler_distr}
\eeq
Here the scatter matrix parameter $\Q\in \PDH(\pdim)$ is uniquely defined up to a positive scalar.  
Although ACG distribution does not belong to the class of RES distributions, 
it is related to it. Namely, 
 an important property of the elliptical family is that 
 $\z/\|\z\|$  has ACG distribution for any $\z \sim \mathcal E_\pdim(0, \M,g)$.  Note that replacing $\x_{ki}$ with $\x_{ki}/\|\x_{ki}\|$
in \eqref{acg_ll} does not affect its minimizer since it is equivalent to subtracting the term $\frac{p}{n_k}\sum_{i=1}^{n_k}  \log\(\x_{ki}^\top\x_{ki}\)$,
which does not depend on $\Q_k$. Consequently, the distribution of the resulting $M$-estimator of scatter is the same under any elliptical distribution.


\section{Preliminaries on $g$-convexity} \label{sec:gconv}

The optimizaton problems defined by \eqref{eq:Mpooled}, \eqref{eq:kpenfun} and \eqref{eq:penfun} are easiest to handle when the target functions to be minimized are convex. The Gaussian negative log-likelihood \eqref{eq:Lg}, for example, is well known to be strictly convex as a function of the inverse covariance matrix. Unfortunately, the functions 
\eqref{eq:Mpooled}, \eqref{eq:kpenfun} and \eqref{eq:penfun} in general tend not to possess this convexity property. Other notions of convexity, though, can be applied. 
Briefly summarizing, convexity properties of 
sets in metric spaces depend on the definition of the shortest paths (geodesic curves) between pairs of points. 
Thus, when the metric is altered, geodesic curves change and consequently so does the notion of convexity. In our treatment, we use the notion of $g$-convexity relative to the 
intrinsic Riemannian manifold structure of the positive semi-definite cone; see \cite{absil2009optimization} and \cite{wiesel2015structured} for a more detailed exposition. 


The set $\PDH(\pdim)$ can be endowed with a smooth Riemannian manifold structure  by changing the usual Euclidean metric to the Riemannian one. The latter can be defined by stipulating the notion of a {\paino geodesic path} from $\M_0 \in \PDH(\pdim)$ to $\M_1 \in \PDH(\pdim)$ and setting it to be
\begin{align} 
\label{eq:geopath}
\M_t &= \M_0^{1/2}\left(\M_0^{-1/2}\M_1\M_0^{-1/2}\right)^t \M_0^{1/2}, \quad t \in [0,1].
\end{align}
Given $\M_0, \M_1 \in \PDH(\pdim)$, we have $\M_t \in \PDH(\pdim)$ for $0 \leqslant t \leqslant 1$, therefore, $\PDH(\pdim)$ is a {\paino geodesically convex set}. A function $h: \PDH(\pdim) \rightarrow \R$ is a {\paino $g$-convex function} if
\beq \label{eq:gconfun}
h(\M_t) \leqslant (1-t)~h(\M_0) + t~h(\M_1),\;\; t \in (0,1). 
\eeq
If the inequality is strict, $h$ is said to be strictly $g$-convex. When $\pdim =1$, $g$-convexity/strict $g$-convexity is equivalent to the function 
$h(e^x)$ being convex/strictly convex over $\R$, i.e. $h(s)$ is convex/strictly convex in $\log(s)$ for $s > 0$. 
the concept of $g$-convexity enjoys properties similar to those of convexity in Euclidean spaces. In particular, if $h$ is $g$-convex on $\PDH(\pdim)$ then any local minimum is a global minimum. Furthermore, if a minimum is obtained in $\PDH(\pdim)$ then the set of all minima form a $g$-convex subset of $\PDH(\pdim)$. If $h$ is strictly $g$-convex and a minimum is obtained in $\PDH(\pdim)$, then it is unique, see \cite{wiesel2012geodesic,wiesel2015structured} and reference therein for more details. An important additional property of $g$-convexity, 
not shared by convexity in Euclidean spaces, is that if $h(\M)$ is $g$-convex/strict $g$-convex in $\M$, then it is also $g$-convex/strict $g$-convex  in $\M^{-1}$.

Similarly, given a naturally induced manifold structure 
over $\PDH(\pdim) \times\PDH(\pdim)$, we say that a function $u: \PDH(\pdim) \times\PDH(\pdim) \rightarrow \R$ is jointly $g$-convex if 
\[
u(\M_t^*,\M_t) \leqslant (1-t)~u(\M_0^*,\M_0) + t~u(\M_1^*,\M_1) ~ \mbox{ for } ~ t \in (0,1),
\]
where $\M_t^*$ and $\M_t$ are defined as in \eqref{eq:geopath}.

To establish the existence and uniqueness of the solutions to the minimization problems \eqref{eq:Mpooled}, \eqref{eq:kpenfun} and \eqref{eq:penfun}, a basic requirement is 
that loss functions $\mathcal L_k(\M_k)$ in \eqref{eq:Lk} be $g$-convex in $\M_k \in \PDH(\pdim)$.  This is achieved when the respective loss functions 
$\rho_k(t)$, $k=1,\ldots,\kdim$, satisfy Condition~\ref{cond:rho}, see \cite{zhang2013multivariate} or \cite[Lemma~1]{ollila2014regularized}.  Many common loss 
functions satisfy this condition, such as the Gaussian, $t$, Huber and Tyler loss functions  given in Section~\ref{sec:loss}.  Using the terminology of this
section, Condition~\ref{cond:rho} simply requires that $\rho_k(t)$ be a nondecreasing, continuous $g$-convex function.

\section{Distance measures for covariance matrices} 
\label{sec:pencost}  

Optimization problems \eqref{eq:kpenfun} and \eqref{eq:penfun} balance the overall loss between the separate group $M$-estimation with shrinkage towards a mutual 
joint (pooling) center. 
The penalty terms introduced in \eqref{eq:kpenfun} and \eqref{eq:penfun} require specifying a distance function $d(\A,\B): \PDH(\pdim) \times\PDH(\pdim)  \rightarrow \R_0^+$, 
(we sometimes refer to $d$ as a penalty function due to its role here). Listed below are some properties one may desire for a distance function.
\begin{itemize} 
\item[(D1)] $d(\A, \B)=0$ if $\A = \B$,
\item[(D2)] $d(\A,\B)$ is jointly $g$-convex,
\item[(D3)] {\it symmetry}: $d(\A,\B)= d(\B,\A)$. 
\item[(D4)] {\it affine invariance}: $d(\A,\B) = d(\C\A\C^\top,\C\B\C^\top)$ for any nonsingular $\C$.
\item[(D5)] {\it scale invariance}: $d(c_1\A,c_2\B)= d(\A,\B)$  for $c_1,c_2>0$,
\end{itemize}
D1 and D2 are necessary requirements. D1 is an obvious requirement, whereas D2 is needed to guarantee that the optimization problem \eqref{eq:penfun} is $g$-convex. 
Properties D3, D4 and D5 are considered optional. When property D4 holds, the resulting estimators of the scatter matrices are affine equivariant. That is,
if we transform the data $\x_{ki} \rightarrow \C\x_{ki}$ for all $k = 1, \ldots, \kdim$; $i = 1, \ldots, n_k$, then 
$\{\Q_1, \ldots, \Q_\kdim, \Q\}  \rightarrow \{\C\Q_1\C^\top, \ldots, \C\Q_\kdim\C^\top, \C\Q\C^\top \}$. Property D5 is useful if we are primarily interested in
the shape of the scatter matrices, that is, the scatter matrices up to a positive scalar. Some scatter scatter estimators (such as Tyler's $M$-estimator) are 
shape estimators only, and for such shape estimators, D5 is necessary.  Property D5 is also important if the individual covariances are believed to be scaled differently,
and so one may wish to only poll together their shapes but not their overall scale. 

As noted in section \ref{gen_set}, every distance function induces a notion of a mean, defined by \eqref{mean_def}.
For example, when $\PDH(\pdim)$ is treated as being embedded within the space of symmetric matrices, 
it inherits the natural Euclidean distance $d_{{\rm F}}(\A,\B) = \| \A - \B \|_{\rm F}^2$, which is the usual Frobenius metric.
The mean \eqref{mean_def} corresponding to the Euclidean distance $d_{{\rm F}}$ is the weighted {\paino arithmetic mean} 
\begin{equation}
\Q_{{\rm F}}(\bom \pi) = \sum_{k=1}^\kdim \pi_k \Q_k. 
\end{equation}
Unfortunately, the Frobenius metric fails to be jointly $g$-convex function, and so does not fit into our framework. However, following are examples of 
$g$-convex distance functions, which can be used to construct $g$-convex optimization problems.
 
\subsection{ Riemannian distance}

The Riemannian distance, 
\[
\dR(\A,\B) =  \| \log (\A^{-1/2}\B\A^{-1/2})\|_{\rm F}^2,
\] 
is the length of the geodesic curve between $\A$ and $\B$ as defined in \eqref{eq:geopath} and hence it is a natural and widely studied distance between positive definite matrices. See e.g., \cite{bhatia2009positive} and references therein for a comprehensive survey. The Riemannian distance satisfies properties D1, D2, D3 and D4. 
The unique extremum  $\Q_{{\rm R}}({\bom \pi})$ of  \eqref{mean_def} 
is a weighted form of the {\paino Karcher mean}  (or Riemannian or geodesic mean), see \cite{moakher2005differential,bhatia2009positive}.
It was shown in \cite{moakher2005differential} to be the unique positive definite solution of
\begin{align} \label{eq:karcher}
\sum_{k=1}^\kdim \pi_k \log( \Q_k^{-1} \Q_{{\rm R}})= \bo 0.
\end{align}
Equation \eqref{eq:karcher} can be written in different forms, e.g., $\sum_{k=1}^\kdim \pi_k \log( \Q^{1/2}_{{\rm R}} \Q_k^{-1} \Q_{{\rm R}}^{1/2})= \bo 0$, using the formula $\A^{-1} \log(\B)\A = \log(\A^{-1}\B\A)$, 
valid for any invertible matrix $\A$ and any matrix $\B$ having real positive eigenvalues \cite{higham2008functions}.
Equation \eqref{eq:karcher} does not have a closed-form solution and a number of rather complex numerical approaches have been proposed to compute the solution. See e.g.,  \cite{bini2013computing,zhang2013majorization} and references therein for a number of such techniques.

\subsection{ Ellipticity distance } 

The ellipticity distance,
\[
\dE(\A,\B) = \pdim \log \frac{1}{\pdim} \Tr(\A^{-1}\B)-\log|\A^{-1}\B|,
\]
was first introduced for the penalized robust covariance estimation problem in \cite{wiesel2012unified}. Therein, $\B$ corresponds to a fixed shrinkage target shape matrix, and it is shown that $d_{{\rm E}}(\A,\B)$ is $g$-convex in $\A$ for fixed $\B$. We refer $d_{{\rm E}}$ as the {\paino ellipticity distance} since it is related to the ellipticity factor, 
$e(\M)=  \frac{1}{\pdim} \tr(\M)/|\M |^{1/\pdim}$, i.e.\ the ratio of the arithmetic and geometric means of the eigenvalues of $\M$. The factor $e(\M) \geqslant 1$ with equality 
if and only if $\M \propto \I$. Its relationship to $d_{{\rm E}}$ is given by $d_{{\rm E}}(\A,\B) =  \log e(\A^{-1/2} \B \A^{-1/2})$. 

The ellipticity distance is scale invariant, i.e. it satisfies D5. It also satisfies properties D1, D2, and D4. We summarize these properties in the
following proposition, which also characterize its induced mean \eqref{mean_def}. 
The proof follows readily from the joint $g$-convexity of the trace term $\Tr(\A^{-1}\B)$, which is proven in Lemma \ref{tr_conv} of the Appendix.

\begin{proposition}
\label{de_prop}
The ellipticity distance  $d_{{\rm E}}(\A,\B)$ satisfies D1, D2, D4 and D5. Furthermore, the optimization problem \eqref{mean_def} has a unique minimizer (up to a scale),
with the minimizer being the unique solution (up to scale) of the fixed-point equation
\beq \label{eq:Emean}
\Q_{{\rm E}} =  \left( \sum_{k=1}^\kdim \pi_k  \frac{p \Q_k^{-1} }{\Tr(\Q_k^{-1}\Q_{{\rm E}})} \right)^{-1}.
\eeq 
\end{proposition}

The ellipticity induced mean $\Q_{{\rm E}}$ is related to the {\paino harmonic mean}, \cite{bhatia2009positive}, of positive definite matrices.  
In particular, it can be viewed as an (implicitly) weighted harmonic mean of normalized scatter matrices.

\subsection{Kullback-Leilber distance } 

The information theoretic (Gaussian)  Kullback-Leibler (KL) divergence \cite{cover2012elements} is defined as 
\[
\dI(\A,\B) =  \Tr(\A^{-1}\B) - \log|\A^{-1}\B | - \pdim.
\]
In statistics literature it has gained popularity due to the seminal works of James and Stein \cite{james1961estimation} who utilized it in the risk function of covariance matrices. It has also been recently used as shrinkage penalty in covariance estimation problems in \cite{sun2014regularized}, who considered a single sample case with $\B$ played being a fixed shrinkage target matrix. The next claim shows the $g$-convexity of  KL-distance and  provides the respective mean \eqref{mean_def}. 

\begin{proposition}
\label{di_prop}
The KL-distance  $\dI(\A,\B)$ satisfies D1, D2 and D4. Furthermore, the optimization problem \eqref{mean_def} has the unique minimizer 
\beq \label{eq:Imean}
\Q_{\rm KL}(\bom \pi) =  \left( \sum_{k=1}^\kdim \pi_k \Q_k^{-1}   \right)^{-1},
\eeq 
which corresponds to a weighted harmonic mean. 
\end{proposition} 

The metrics $d_{{\rm E}}$ and $\dI$ are thus closely related, however, the former is a scale invariant metric while the latter is not. At the same time $\Q_{{\rm E}}$ is given by an implicit weighted harmonic mean equation \eqref{eq:Emean} while $\Q_{{\rm KL}}$ is given explicitly by \eqref{eq:Imean}. 
Thus due to scale invariance, one looses in the simplicity of calculation. The scale invariance property of the penalty function is especially useful  in problems when the unknown elliptical distributions of the samples are different (heterogeneous environment) and hence the scatter matrices obtained using same loss function would have a priori different scalings. This follows from the fact that any $M$-estimator $\hat \M$ provides an estimate up to a constant and the constant of proportionality depends on the underlying distribution, as well.

\section{Fixed-point algorithms}  \label{sec:uniq} 
 
In this section we propose fixed-point (FP) algorithms for computing the regularized scatter estimators. We first consider Proposal~2, which corresponds to optimization 
problem \eqref{eq:penfun}. If we differentiate \eqref{eq:penfun} with respect to $\Q_k^{-1}, k = 1, \ldots, K$ and $\Q$, we obtain the first order optimality conditions 
on the extremum, i.e.\ the $M$-estimating equations, which are
\begin{align}
\frac{\partial \mathcal{L}_k (\Q_k)}{\partial \Q_k^{-1}} + \lambda \frac{\partial d(\Q_k,\Q)}{\partial \Q_k^{-1}} &= \bo 0,\quad k=1,\dots,K, \\
\sum_{k=1}^K\pi_k  \frac{ \partial d(\Q_k,\Q)}{\partial \Q} &= \bo 0, \label{mean_dd}.
\end{align}
If we denote $\partial d(\Q_k,\Q)/\partial \Q_k^{-1}$ by $\S(\Q_k,\Q)$ and use the form of $\mathcal L_k (\Q_k)$ given in \eqref{eq:Lk}, then 
the first equation becomes
\begin{equation} \label{eq:FPk} 
\Q_k = \frac{1}{n_k} \sum_{i=1}^{n_k}  u_k(\x_{ki}^\hop \Q_k^{-1}\x_{ki}) \x_{ki}\x_{ki}^\hop +  \lam \S(\Q_k,\Q),
\end{equation}
with $u_k(t) = \rho_k^\prime(t)$, $k = 1,\ldots,\kdim$, acting as weight functions. 
The second equation \eqref{mean_dd} coincides with the definition of the mean given in \eqref{mean_def}. Hence, a general FP algorithm for finding the solution to the optimization
problem \eqref{eq:penfun} is given by the iterative scheme  
\begin{align}  
\Q_k &\leftarrow \frac{1}{n_k} \sum_{i=1}^{n_k}  u_k(\x_{ki}^\hop \Q_k^{-1}\x_{ki}) \x_{ki}\x_{ki}^\hop + \lam \S(\Q_k,\Q)  \label{eq:first_it}, \quad k=1,\ldots,\kdim \\
\Q &\leftarrow \Q_{{k}}(\bom \pi), \label{eq:sec_it}
\end{align}
which updates the covariance matrices in natural order, i.e.,  from $\Q_{1}, \ldots \Q_{\kdim}$ to $\Q$, and cyclically repeating the procedure until
convergence.  One may view it as blockwise alternating minimization algorithm in which one minimizes the objective function in one block at a time while keeping  others fixed at their current iterates.
The convergence properties of such a scheme is omitted and is a subject of a follow up paper (under preparation)   
in which blockwise minimization majorization (MM) algorithmic scheme \citep{hunter2004tutorial,razaviyayn2013unified} is utilized for proving convergence. We note that an MM algorithm (in the single covariance estimation problem, $\kdim=1$) have been recently used in  \cite{sun2014regularized} and \cite{wiesel2015structured} in constructing simple convergence proof of FP algorithm of regularized Tyler's $M$-estimator.  At this point it simply suffices to say that in practice, when the scheme converges, then, by $g$-convexity (or strict $g$-convexity), it 
must converge to a solution (or the unique solution) to the optimization problem \eqref{eq:penfun}.

For the first proposal, which corresponds to the optimization problem \eqref{eq:Mpooled}-\eqref{eq:kpenfun}, one first needs to compute $\Q$, i.e.\ the optimizer 
of \eqref{eq:Mpooled}. This simply involves computing a non-penalized $M$-estimator of scatter,  for which computational algorithms have been well studied, see e.g., 
\cite{kent1991redescending,dumbgen2016new}. A simple FP algorithm for $\Q$ is given by the iterative scheme
\[  \Q \leftarrow \frac{1}{N} \sum_{k=1}^\kdim \sum_{i=1}^{n_k} u_k(\x_{ki}^\hop \Q^{-1}\x_{ki}) \x_{ki}\x_{ki}^\hop.  \]
Given this value of $\Q$, the FP algorithm for the $\Q_k$'s, the optimizers of \eqref{eq:kpenfun}, corresponds to \eqref{eq:first_it} with $\Q$ held fixed.

The exact forms of the derivatives for $\dR, \dE$ and $\dI$ are respectively:
\begin{align*}
\text{\scriptsize $\bullet$}& \quad \  \dR'(\M_k,\M)  \ 
= 2 \log(\M\M_k^{-1})\M_k, \\ 
\text{\scriptsize $\bullet$}& \quad \ \dE'(\M_k,\M) \ \,  = \frac{\pdim\M}{\Tr(\M_k^{-1}\M)} - \M_k, \\
\text{\scriptsize $\bullet$}& \quad \dI'(\M_k,\M)  \, = \M-\M_k.
\end{align*} 
The specific form of the fixed point (FP) algorithm which utilizing each of these distances is given below. 
For simplicity, let
\[
\bom \Psi_k(\Q_k)=  \frac{1}{n_k} \sum_{i=1}^{n_k}  u_k(\x_{ki}^\hop \Q_k^{-1}\x_{ki}) \x_{ki}\x_{ki}^\hop.
\]
Also, map $\lambda \ge 0$ to $\beta = 1/(1+\lambda) \in (0,1]$. Note that $\be$ is a regularization parameter given in the formulation \eqref{eq:kpenfun2}. 
\begin{itemize}
\item {\bf{ Ellipticity distance $\dE(\A,\B)$}}:
\begin{align}  
\Q_k &\leftarrow   \beta \bom \Psi_k(\Q_k) + (1-\beta)
\frac{ p\Q }{\Tr( \Q_k^{-1}\Q)} ,   \label{eq:FP1_dE} \\
\Q &\leftarrow  \left( \sum_{k=1}^\kdim \pi_k  \frac{p \Q_k^{-1} }{\Tr(\Q_k^{-1}\Q)} \right)^{-1}. \label{eq:FP2_dE}
\end{align} 
Note that in this case the first equation is not sensitive to the scaling of $\Q$, as we expected from the scaling invariance properties of the penalty $\dE$.

\item {\bf KL-distance $\dI(\A,\B)$}:
\begin{align} 
 \Q_k &\leftarrow \beta \bom  \Psi_k(\Q_k) 
 +  (1- \beta) \Q, \label{kl:fixedp} \\
\Q &\leftarrow  \left( \sum_{k=1}^\kdim \pi_k \Q_k^{-1}   \right)^{-1}. \label{kl:fixedp2}
\end{align}
It is shown by Theorems 1 and 2 in \cite{ollila2014regularized} that if the loss function $\rho_k$ is bounded below, then for any fixed $\Q$ the
FP algorithm \eqref{kl:fixedp} always converges to a unique solution.

\item {\bf{Riemannian distance $\dR(\A,\B)$}}:
\begin{align}
\Q_k &\leftarrow \left[ \I - 2 \lambda  \log(\M\M_k^{-1}) \right]^{-1}\bom  \Psi_k(\Q_k),  \label{r:fixedp}\ \\
\Q &\leftarrow \Q_{\rm R}(\bom \pi),
\end{align}
where $\Q_{\rm R}(\bom \pi)$ is the solution  to \eqref{eq:karcher}
\end{itemize}

\begin{remark} 
Note that iterative algorithms for KL and ellipticity metrics provide simple FP algorithms, but the Riemannian metric does not admit a simple FP equation for the joint center $\M$ update, but rather requires more complex schemes; see e.g. \cite{zhang2013majorization} for an appropriate iterative algorithm. In addition, note that the last step requires solving $\Q_{{\rm R}}(\bom \pi)$ as a solution to \eqref{eq:karcher} which is computationally demanding task. Also the updates for $\M_k$, $k=1,\ldots,\kdim$ in  \eqref{r:fixedp} are computationally more demanding than the updates  \eqref{kl:fixedp} and \eqref{eq:FP1_dE} corresponding to the other penalties. Therefore, since Riemannian distance requires a more specialized algorithm, we do not consider the Riemannian penalty further and exclude it from our simulation studies. 
\end{remark} 


\subsection{Tyler's loss function and the ellipticity penalty} \label{sec:tyl}

In this section we treat Tyler's loss function $\rho_k(t)=p \log t,$ for $k=1,\ldots,\kdim$, in more detail. 
For this case the $M$-estimation loss function  \eqref{eq:Lk} is scale invariant, i.e., $\mathcal L_k( c\Q_k)= \mathcal L_k(\Q_k)$,
Hence we will mainly consider the scale invariant distance $\dE(\A,\B)$ when using Tyler's loss function. Here, 
the fixed-point iteration in \eqref{eq:FP1_dE} becomes 
\begin{align} 
 \Q_k &\leftarrow \beta \frac{p}{n_k}\sum_{k=1}^\ndim \frac{\x_{ik}\x_{ik}^\hop}{\x_{ik}^\hop\Q_k^{-1}\x_{ik}}+ (1- \beta)\frac{p}{\Tr( \Q_k^{-1}\Q)}\Q, \label{tyl:fixedp}
\end{align}
with the joint center update being the same as in \eqref{eq:FP2_dE}. The resulting estimators $\Q_1, \ldots, \Q_K$ and $\Q$ are well defined only up to their
shape. That is, $\{\Q_1, \ldots, \Q_K, \Q\}$ is a solution if and only if $\{\sigma_1\Q_1, \ldots, \sigma_K\Q_K, \sigma\Q\}$ is a solution for any
positive $\sigma_1, \ldots, \sigma_K$ and $\sigma$. 

Curiously, if we choose a solution to \eqref{tyl:fixedp} and \eqref{eq:FP2_dE} for which $\Tr( \Q_k^{-1}\Q) = p$, then we see that this also
gives the solution to \eqref{kl:fixedp} and \eqref{kl:fixedp2}, i.e.\ when using Tyler's loss function with KL-penalization. As noted previously, for fixed $\Q$,
a unique solution to \eqref{kl:fixedp} always exists whenever the loss function $\rho_k$ is bounded below. Tyler's loss function, though, is 
not bounded below, and additional condition are needed to ensure existence. In particular, as a corollary to Theorem 4 in \cite{ollila2014regularized}
we have the following result. 

\begin{theorem} \label{t:T} For Tyler's loss function $\rho_k(t)=p \log t$ and $0 \leqslant \beta=1/(1+\lambda) < 1$,  
a necessary condition for program \eqref{eq:penfun} to have a non-singular minimum is that for each of the $k=1,\ldots, K$ group samples the inequality
\beq \label{eq:condA}
P_{n_k,k}(\mathcal V)  =  \frac{ \# \{ \z_{ik} \in \mathcal V \}}{n_k}  <  \frac{\mathrm{dim}(\mV)}{p\beta}
\eeq 
holds for any subspace $\mV$  of $\mathbb{R}^{\pdim}$. Furthermore, if we replace the '$<$' with '$\le$' in Condition \eqref{eq:condA}, then
this becomes a sufficient condition for ensuring, for fixed $\Q$, that \eqref{tyl:fixedp} admit unique solutions for $\Q_k$ up to a scale.
\end{theorem}

 Here, $P_{n_k,k}$ is the empirical measure for the $k$-th group sample. Condition \eqref{eq:condA} 
implies that when the data is in general position for each group sample, we need $n_k > p \beta $.



\section{Cross validation procedure} \label{sec:CV}

Let us  describe a simple cross validation (CV) procedure that can be utilized for penalty parameter selection  $\be \in (0,1)$. 
Recall that the objective function 
given all the data for the parameters $\M_1, \ldots, \M_K$ is
\[
 \mathcal{L}(\M_1, \ldots,\M_K) = \sum_{k=1}^\kdim \Bigg \{ \sum_{i=1}^{n_k} \rho_k(\x_{ki}^\top \M_k^{-1}\x_{ki}) - n_k  \log|\M_k^{-1}|  \Bigg \}.
 \] 
Partition each data set $\Z_k=\{ \x_{k1},\ldots, \x_{k \ndim_k}\}$ into $\qdim$ separate sets of approximately similar size (or exactly equal size when $\mathrm{mod}(\ndim_k,\qdim)=0$), i.e.,   let 
$
I_{k1} \cup I_{k2} \cup \cdots  \cup I_{k\qdim} = \{1,\ldots,\ndim_k\} \equiv [n_k]
$
denote the indices of $\qdim$ data folds of the $k$th data set. 
Common choises are $\qdim=5, 10$  or $\qdim=n_k$, which is known as {\it leave-one-out cross validation}.  
When  we leave $q$th fold out from the $k$th data set $\Z_k$, we obtain a reduced data set, denoted by 
data set $\Z_{-q,k}$, that does not include the observations  $\{\x_{kq}\}$, $q  \in  [ n_k ] \setminus I_{kq}$, in the $q$th fold.  
Cross validation scheme then  proceeds as follows: 

\begin{enumerate} 

\item {\bf for}  $\beta \in [\beta]$ ($=$ a grid of $\beta$ values in $(0,1)$) and $q \in \{1,\ldots,\qdim\}$ {\bf do} 
\begin{itemize}
\item Compute $\big\{ \hat \M_k(\be,q) \big\}_{k=1}^\kdim$ based on  the data sets $\big\{  \Z_{-q,k} \big\}_{k=1}^\kdim$. 
 
\item {\paino CV fit} for $\be$ is computed over the  $q$th folds that were left out: 
\beq \label{eq:CV} 
\mathrm{CV}(\be,q) =   \sum_{k=1}^\kdim  \bigg \{   \sum_{ \tilde q  \in I_{kq} } \rho_k \big( \x_{k \tilde q}^\top \big[\hat \M_k(\be,q) \big]^{-1} \x_{k \tilde q} \big)   -  (\#I _{kq} ) \cdot  \log \big| \hat \M_k(\be,q)^{-1} \big|  \bigg \}
\eeq 
where  $ \# I _{kq} $ denotes the cardinality of set $I_{kq}$.   
\end{itemize} 
\item[] {\bf end}

\item Compute the average CV fit:  $\mathrm{CV}(\beta)= (1/\qdim) \sum_{q=1}^\qdim \mbox{CV}(\beta,q) $, $\forall \be \in [\be]$. 

\item Select $\hat \be_{\mathrm{CV}} =\arg \min_{\be \in [\be]} \mathrm{CV}(\be)$.
\item Compute  $\big\{ \hat \M_k(\be) \big\}_{k=1}^\kdim$ based on the entire data sets $\{\Z_k\}_{k=1}^\kdim$ for $\be=\hat \be_{\mathrm{CV}}$.

\end{enumerate} 

It is easy to imagine a variant of this approach in which definition of CV fit  is tuned towards a measure that arises from application perspective. 
For example, in discriminant analysis  described in  Section~\ref{sec:RDA}  one may wish to   
replace the CV fit measure in \eqref{eq:CV} by the classification error rate over the $q$th folds. This approach, however,  gave essentially same results, and hence the CV scheme described above is used also in this setting due to its simplicity.

\section{Regularized discriminant analysis (RDA)}\label{sec:RDA} 

The classic Fisher's QDA is based on the assumption that each class contain a sample of i.i.d. random vectors from the $\pdim$-variate Gaussian distribution with mean vector $\bom \mu_k$ and covariance matrix $\M_k$. For simplicity of exposition we assume that the class prior probabilities are equal. 
The QDA classification rule then assigns  a new measurement $\z$ to a group $\hat k$, where 
\beq \label{eq:QDA}
\hat k = \arg \min_{1 \leq k \leq \kdim}  \big\{ (\z-\bom \mu_k)^\top \M_k^{-1}( \z-\bom \mu_k) + \ln | \M_k | \big \}.
\eeq 
If  all class covariance matrices are presumed to be identical, i.e., $\M_k = \M$ for $k=1,\ldots,\kdim$,  
then the rule simplifies to 
$\hat k = \arg \min_{1 \leq k \leq \kdim}  (\z-\bom \mu_k)^\top \M^{-1}( \z-\bom \mu_k)$,  
referred to as LDA rule hereafter. 
In general, QDA or LDA perform well when the class distributions are approximately normal and good estimates based on the training data can be obtained for the population parameters, mean vectors $\bmu_k$ and covariance matrices $\M_k$. 
These are usually estimated by the sample mean vectors, $\bar \z_k= (1/n_k) \sum_{i=1}^{n_k} \z_{ki}$, and sample covariance matrices
$\bo S_k = \frac{1}{n_k} \sum_{i=1}^{n_k} (\z_{ki} - \bar \z_k)  (\z_{ki} - \bar \z_k)^\top$    
of the training samples $\Z_k$, $k=1,\ldots,\kdim$.  
QDA generally requires larger sample sizes than LDA and is often reported to be more sensitive to violations of the assumptions. QDA  also can not be applied if $\ndim_k \leq \pdim$ for any class and may exhibit poor performance when $\ndim_k$ is not considerably larger than the dimension $\pdim$. 
  LDA has the benefit of requiring only that $N= \sum_{k=1}^\kdim n_k > \pdim$. 
 LDA can be  viewed as a form of regularized QDA that decreases the variance by using a pooled covariance matrix estimate, 
$\bo S = \sum_{k=1}^\kdim  \pi_k \bo S_k$. 
This can sometimes lead to superior performance compared to QDA 
especially in small-sample settings even if the population class covariance matrices are substantially different. 

 The idea in RDA proposed in \cite{friedman1989regularized} is to  replace the unknown 
covariance matrices $\M_k$ in the QDA rule \eqref{eq:QDA} by shrinkage estimates $\bo S_k(\be)$ 
defined in \eqref{eq:Friedman},
where $\be \in [0,1]$ denotes the shrinkage regularization parameter. 
 If $\be=1$, then one obtains the conventional empirical QDA rule and if $\be=0$, then one obtains the empirical LDA rule based on the pooled sample covariance matrix $\bo S$.  
A value $\be \in (0,1)$, between these two extremes, then offers a compromise between  LDA and QDA.
In our RDA approach we use the developed robust estimators $\hat \M_k(\be)$ instead of the shrinkage sample covariance matrices $\bo S_k(\be)$. 
The RDA rule becomes
\beq \label{eq:RDA}
\hat k = \arg \min_{1 \leq k \leq \kdim}  \big\{ (\z- \hat{\bom \mu}_k)^\top [\hat \M_k(\be)]^{-1}( \z-\hat{\bom \mu}_k) + \ln |\hat  \M_k(\be) | \big \}.
\eeq 
For robust loss functions, we employ the spatial median \cite{brown1983statistical} as  an estimate $\hat{\bom \mu}_k$ of location, whereas sample mean is used for Gaussian loss function.  Note that the shrinkage scatter matrix estimators $\hat \M(\be)$ are computed  using the centered data. 



\subsection{Simulation set-up} \label{sec:simul}

Population class conditional distributions are chosen to be $\pdim$-variate elliptical distributions and the total sample size is fixed to $N= \sum_{k=1}^\kdim n_k = 100$, the number of groups is $\kdim = 3, 5$   
and  dimension varies from $\pdim = 10, 20, 30$.  
 For simplicity we use the same loss function $\rho=\rho_k$ for each $\kdim$ samples\footnote{Using different loss functions for different clusters can be advisable when some {\it a priori} information is available about the class distributions.}. The class distributions  follow  Gaussian distributions or $\pdim$-variate heavy-tailed $t_\nu$-distributions with  $\nu=2$  degrees of freedom.  

Both the Proposal 1 and Proposal 2 can be used to estimate the regularized class scatter matrices $\hat \M_k(\be)$ that are needed in RDA rule.   
We use notation Prop1$(\rho,d)$ and Prop2$(\rho,d)$,  where $\rho$ refers to the used loss function and $d$ to the used distance function. 
To identify the used loss function $\rho$, we use letters G, H, and T to refer to Gaussian loss 
$\rho_G$, Huber's loss $\rho_{\rm H}$ in \eqref{eq:huber_rho} and Tyler's loss function $\rho_{\rm T}$,  
respectively. Furthermore,  E and KL  indicate that ellipticity distance $\dE$ and KL-distance $\dI$, respectively,  are chosen as the distance function $d$. With the above notation,  Prop(G,KL) then refers to original 
RDA rule based on $\bo S_k(\be)$ in \eqref{eq:Friedman} and Prop2(T,E), for example,  indicates that  Tyler's loss function $\rho_{\rm T}$ and ellipticity distance $\dE$  are used when estimating the scatter matrices using Proposal 2.   
For Huber's loss function we used $c^2 = F^{-1}_{\chi^2_\pdim}(0.9)$ as the tuning threshold $c$. 

We compute the estimated misclassification risk as follows. 
The sample lengths follow multinomial distribution  $(n_1,\ldots,\ndim_\kdim) \sim \mathrm{Multin}(N, \bo p)$, where the class probabilities are $p_1=p_2=1/4$ and $p_3=1/2$ when $\kdim=3$ and  $p_1=p_2=p_3=1/6$, and $p_4=p_5=1/4$ when $\kdim=5$. 
Then random vectors were drawn from the appropriate class distributions. 
Each such {\it training data} set was used to construct the estimated 
discriminant rules. An additional {\it test data} set of same sample lengths $n_i$-s as the training data was  generated 
 and classified with the discriminant rules derived from the training set, thereby yielding an estimate of the misclassification risk. 
For  RDA we report the  misclassification risk based on the best  value of shrinkage parameter $\be$. For each MC trial we compute the RDA rule \eqref{eq:RDA} for $\be$ in the grid 
$
[B]=[0.01,0.03,\ldots,0.49, 0.55, 0.60, \ldots, 0.9]  
$
 and the respective estimated misclassification risk. The best value $\be_0 \in [B]$ is chosen for each RDA approach  as the smallest value in the grid $[B]$ that produced the minimum misclassification risk.   Reported results are averages over 300 MC trials.  

We compare the performance of RDA approaches to conventional LDA and QDA rules as well as to  Oracle estimators. We use notation Oracle1 to refer to QDA rule in \eqref{eq:QDA} 
that uses both the true mean vectors $\bom \mu_k$ and true scatter matrices $\M_k$.  Oracle2 denotes QDA rule that uses the true scatter matrices $\M_k$, but estimated 
mean vectors $\hat {\bom \mu}_k$. For Gaussian samples, $\hat {\bom \mu}_k$ used in Oracle2 are the sample mean vectors and for $t_2$-distributed samples, they $\hat {\bom \mu}_k$ are the spatial medians of the samples.


\subsection{Simulation results}

\begin{table}[!t]
\vspace{-0.5cm}
\caption{ Average (\%)  test misclassification errors for unequal spherical covariance matrices ($\M_k=k \I$) for Gaussian (upper table) and $t_2$-distributed (lower table) clusters.  The quantities in subscript inside the parantheses are the standard deviations.} \label{tab:caseC}
\begin{center}
{\small
\begin{tabular}{|  l | c |  c |  c |  c |  c |  c |}
\hline
    &\multicolumn{3}{| c|}{$\kdim=3$} & \multicolumn{3}{| c|}{$\kdim=5$} \\ \hline 
  method 		&   $p=10$  		&  $p=20$      		&  $p=30$ 		&   $p=10$   		&  $p=20$      		&  $p=30$   \\ \hline 
Oracle1             & 8.8$_{(2.6)}$ 	& 6.2$_{(2.3)}$		& 4.6$_{(1.9)}$	 	&  9.4$_{(2.8)}$		& 7.7$_{(2.8)}$		& 6.2$_{(2.3)}$ \\ 
Oracle2		& 9.8$_{(3.1)}$		& 7.6$_{(2.6)}$		& 6.0$_{(2.3)}$		&  11.3$_{(2.9)}$	& 10.1$_{(3.3)}$	& 9.2$_{(2.9)}$\\
QDA           	& 19.9$_{(4.4)}$	& $-$	  		& $-$			&   $-$ 			&  $-$ 			&  $-$ 	\\
LDA           	& 17.1$_{(3.8)}$	& 20.5$_{(4.3)}$	& 24.0$_{(4.9)}$ 	& 18.8$_{(3.8)}$	& 24.2$_{(4.7)}$	& 29.0$_{(5.0)}$\\
Prop1(G,KL)	& 12.2$_{(3.1)}$	& 14.6$_{(3.5)}$	& 17.9$_{(4.3)}$	& 15.4$_{(3.4)}$	& 20.5$_{(4.1)}$ 	& 25.8$_{(4.8)}$\\
Prop1(H,KL)	& 12.4$_{(3.2)}$	& 14.6$_{(3.5)}$	& 17.7$_{(4.1)}$	& 15.4$_{(3.3)}$	& 20.3$_{(4.1)}$ 	& 25.5$_{(4.8)}$\\
Prop1(T,E)	& 10.9$_{(3.1)}$ 	& 12.1$_{(3.3)}$	& 16.5$_{(3.9)}$ 	& 13.5$_{(3.4)}$	&17.1$_{(4.3)}$ 	& 23.9$_{(5.1)}$\\ 
Prop2(G,E)    	& 10.5$_{(3.0)}$ 	& 11.5$_{(3.3)}$	& 15.9$_{(3.8)}$	& 12.9$_{(3.4)}$	&16.5$_{(4.0)}$ 	& 22.7$_{(4.8)}$ \\ 
Prop2(T,E)  	& 10.9$_{(3.1)}$ 	&12.1$_{(3.3)}$		& 16.5$_{(3.9)}$	& 13.5$_{(3.4)}$	&17.1$_{(4.3)}$ 	& 23.9$_{(5.1)}$\\
Prop2(H,E)	& 10.5$_{(3.0)}$ 	&  11.6$_{(3.3)}$	& 15.7$_{(3.8)}$	& 12.9$_{(3.3)}$	&16.5$_{(4.1)}$ 	& 22.6$_{(4.8)}$\\ 
Prop2(H,KL)	&  12.3$_{(3.2)}$ 	&  14.8$_{(3.6)}$	& 18.0$_{(4.1)}$	& 15.2$_{(3.4)}$	& 20.1$_{(4.2)}$ 	& 25.4$_{(4.7)}$ \\ \hline 

Oracle1               & 15.7$_{(3.8)}$ 	& 18.2$_{(3.9)}$	& 21.1$_{(4.0)}$ 	& 20.8$_{(4.1)}$ 	& 24.5$_{(4.2)}$	& 27.8$_{(4.6)}$ \\ 
Oracle2		& 16.2$_{(3.5)}$	& 19.1$_{(4.2)}$     	& 21.9$_{(4.1)}$ 	& 21.7$_{(4.3)}$	& 25.8$_{(4.1)}$ 	& 29.1$_{(4.6)}$\\
QDA           	& 26.9$_{(5.2)}$	& $-$	 		& $-$			&$-$				& $-$			&  $-$		 \\
LDA           	& 21.8$_{(4.9)}$	& 25.3$_{(5.3)}$	& 27.7$_{(5.3)}$	& 28.6$_{(5.6)}$	& 32.9$_{(5.6)}$ 	& 36.2$_{(5.4)}$ \\
Prop1(G,KL)	& 19.7$_{(4.8)}$	&22.7$_{(5.2)}$		& 24.7$_{(5.1)}$	& 27.2$_{(5.7)}$ 	& 31.0$_{(5.3)}$ 	& 33.8$_{(5.4)}$ \\
Prop1(H,KL)     & 15.5$_{(3.7)}$	&17.9$_{(4.0)}$		& 20.3$_{(4.1)}$ 	& 21.0$_{(4.1)}$	& 24.6$_{(4.5)}$	& 28.2$_{(4.6)}$ \\
Prop1(T,E)	& 16.8$_{(4.0)}$	&20.4$_{(4.3)}$		& 23.4$_{(4.7)}$	& 23.7$_{(4.5)}$ 	& 29.6$_{(5.0)}$	& 34.0$_{(5.3)}$ \\ 
Prop2(G,E)    	& 22.3$_{(5.9)}$	&24.3$_{(5.1)}$		& 25.9$_{(4.8)}$	& 28.1$_{(5.4)}$	& 32.5$_{(5.4)}$ 	& 35.4$_{(5.1)}$\\
Prop2(T,E)  	& 16.8$_{(4.0)}$	&20.4$_{(4.4)}$		& 23.5$_{(4.8)}$	& 23.7$_{(4.5)}$	& 29.7$_{(5.0)}$	& 34.1$_{(5.3)}$ \\
Prop2(H,E)	& 16.6$_{(3.9)}$	&20.2$_{(4.4)}$		& 23.6$_{(4.6)}$	& 23.1$_{(4.5)}$	& 29.1$_{(4.7)}$	& 33.8$_{(5.3)}$ \\
Prop2(H,KL)	& 15.5$_{(3.7)}$	&17.9$_{(4.0)}$		& 20.5$_{(4.1)}$	& 21.0$_{(4.1)}$	& 24.6$_{(4.4)}$	& 28.2$_{(4.5)}$ \\  \hline 
  \end{tabular}
}
\end{center} 
\end{table}

We consider the case of unequal spherical covariance matrices, where the scatter matrix for the $k$th class is $\M_k= k \I$  for $k=1,\ldots,\kdim$. 
This setting is thus somewhat more favourable to QDA, but due to small sample sizes, the performance of QDA does not exceed that of 
LDA as is shown in Table~\ref{tab:caseC}, which summarizes the simulation results for both the Gaussian and $t_2$ distributions of the classes.  
The symmetry center $\bom \mu_1$ of the first class was the origin and for the remaining classes  $\bom \mu_k$ were taken to have norm equal to $\delta_k=\| \bmu_k\| = 3+k$ in  orthogonal directions  for Gaussian classes and $\delta_k=\| \bmu_k\| = 4+k$ for $t_2$-distributed classes ($k=2,\ldots,\kdim$). 
In the Gaussian case,  Prop2(G,E) and Prop2(H,E)  are offering consistently the best performance, also outperforming  Friedman's original RDA rule, Prop1(G,KL). For example, 
when $\pdim=20$, Prop2(G,E)  offers 3\% improvement in error rate compared to Prop1(G,KL).  This illustrates the benefits of choosing the correct penalty (and hence the estimate of joint center covariance matrix): Prop1(G,KL) and Prop2(G,E) are both using the optimal Gaussian loss function, but can have  4\% (e.g., case $\pdim=10$, $\kdim=5)$ difference in the misclassification rate in favor of  Prop2(G,E.  For Gaussian class distributions, the scale invariant penalty $\dE$ offers the best performance.  
In $t_2$-case, the results illustrate that robust RDA approaches 
  provide  significantly  better misclassification rates  compared to  non-robust RDA approaches using the Gaussian loss function. 
  For example, the best performing robust RDA rules, Prop1(H,KL) and Prop2(H,KL) offer consistently 4--6\% improvements to Friedman's Prop1(G,KL). 
It is somewhat suprising that for $t_2$-distributed samples, KL-distance is generally performing better than the ellipticity distance.   
 Among RDA  approaches, Prop2(G,E) has the worst performance when the class distributions follow the heavy-tailed $t_2$-distribution.

\begin{table}[!t]
\vspace{-0.5cm}
\caption{ Average (\%)  test misclassification errors for identical spherical ($\M_k=\I$) covariance matrices in Gaussian (upper table) and $t_2$-distributed (lower table) samples.  The quantities in subscript inside the parantheses are the standard deviations.  
} \label{tab:caseA}
\begin{center}
{\small
\begin{tabular}{|  l | c |  c |  c |  c |  c |  c |}
\hline
    &\multicolumn{3}{| c|}{$\kdim=3$} & \multicolumn{3}{| c|}{$\kdim=5$} \\ \hline 
  method 		&   $p=10$   		   	&  $p=20$      		&  $p=30$ 		&   $p=10$  		&  $p=20$      		&  $p=30$   \\ \hline
Oracle1             & 8.9$_{(2.9)}$		& 9.2$_{(3.1)}$		&  8.8$_{(2.8)}$		&  10.9$_{(3.2)}$	& 10.9$_{(3.2)}$	&  10.9$_{(3.0)}$\\
Oracle2		& 9.9$_{(3.1)}$		& 10.9$_{(3.3)}$	& 11.2$_{(3.2)}$	& 12.9$_{(3.2)}$	&14.3$_{(3.9)}$		&  15.5$_{(3.6)}$ \\
QDA           	& 18.1$_{(4.1)}$ 	&  $-$			& $-$			&  $-$			&  $-$			&  $-$		 \\ 
LDA           	& 11.3$_{(3.0)}$	& 14.1$_{(3.8)}$	& 16.9$_{(4.2)}$	& 14.6$_{(3.6)}$	& 18.5$_{(4.3)}$	& 22.9$_{(4.7)}$ \\
Prop1(G,KL)	& 10.3$_{(2.9)}$	& 13.0$_{(3.6)}$	& 15.4$_{(3.9)}$	& 13.6$_{(3.4)}$	& 17.4$_{(4.1)}$	&21.7$_{(4.5)}$ \\
Prop1(H,I)        & 10.4$_{(3.0)}$	& 13.0$_{(3.6)}$	& 15.4$_{(4.0)}$	& 13.7$_{(3.4)}$	& 17.5$_{(4.1)}$	&21.8$_{(4.5)}$ \\
Prop1(T,E)	& 10.8$_{(3.1)}$	& 13.4$_{(3.7)}$	& 15.7$_{(3.8)}$	& 14.5$_{(3.4)}$	& 18.3$_{(4.4)}$	&22.9$_{(4.9)}$ \\
Prop2(G,E)    	& 10.3$_{(3.0)}$	& 12.9$_{(3.7)}$	& 15.3$_{(3.7)}$	& 13.9$_{(3.4)}$	& 17.6$_{(4.1)}$	& 22.1$_{(4.7)}$ \\
Prop2(T,E)  	& 10.9$_{(3.1)}$	& 13.4$_{(3.7)}$	& 15.6$_{(3.8)}$	& 14.5$_{(3.5)}$	& 18.2$_{(4.4)}$	&22.9$_{(4.8)}$ \\
Prop2(H,E)	& 10.4$_{(3.0)}$	& 13.0$_{(3.7)}$	& 15.3$_{(3.8)}$	& 14.0$_{(3.4)}$	& 17.7$_{(4.2)}$	& 22.1$_{(4.8)}$\\
Prop2(H,I)	& 10.4$_{(3.0)}$	& 13.0$_{(3.6)}$	& 15.6$_{(4.0)}$	& 13.7$_{(3.4)}$	& 17.5$_{(4.1)}$	&21.8$_{(4.6)}$ \\ \hline
Oracle1               & 12.3$_{(3.2)}$	&12.5$_{(3.5)}$		& 12.0$_{(3.2)}$	& 15.4$_{(3.5)}$ 	& 15.4$_{(3.4)}$	& 15.5$_{(3.3)}$\\
Oracle2		&12.7$_{(3.1)}$		&13.3$_{(3.6)}$		&13.5$_{(3.4)}$		& 16.5$_{(3.7)}$	& 17.3$_{(3.5)}$	& 18.1$_{(3.8)}$\\
QDA           	& 24.4$_{(5.1)}$	& $-$			& $-$			&  $-$			&  $-$			&  $-$		 \\ 
LDA           	& 16.2$_{(4.1)}$	&19.2$_{(4.5)}$		& 21.0$_{(4.7)}$	& 22.1$_{(4.7)}$ 	& 25.2$_{(5.0)}$ 	& 28.1$_{(4.9)}$ \\
Prop1(G,KL)    	& 14.7$_{(4.0)}$	&17.6$_{(4.4)}$		& 19.2$_{(4.3)}$	& 20.7$_{(4.5)}$ 	& 23.7$_{(4.6)}$ 	& 26.6$_{(4.8)}$ \\
Prop1(H,KL)    	& 12.7$_{(3.3)}$	&14.6$_{(3.8)}$		& 16.7$_{(3.8)}$	& 16.9$_{(3.8)}$	& 19.5$_{(3.7)}$ 	& 22.7$_{(4.1)}$ \\
Prop1(T,E)    	& 14.9$_{(4.2)}$	&18.2$_{(5.2)}$		& 21.0$_{(5.7)}$	&20.8$_{(4.7)}$ 	& 26.0$_{(5.3)}$ 	& 31.0$_{(6.1)}$ \\
Prop2(G,E)   	& 17.5$_{(4.9)}$	& 20.1$_{(5.4)}$	& 22.5$_{(5.0)}$	& 23.6$_{(5.2)}$	& 28.4$_{(5.5)}$ 	& 32.1$_{(5.6)}$ \\
Prop2(T,E)  	& 14.9$_{(4.2)}$	& 18.3$_{(5.2)}$	& 21.1$_{(5.7)}$	& 20.8$_{(4.7)}$	& 26.1$_{(5.3)}$ 	& 31.1$_{(6.2)}$ \\
Prop2(H,E)	& 14.5$_{(4.0)}$	& 17.5$_{(4.9)}$	& 20.4$_{(5.2)}$	& 19.9$_{(4.5)}$	& 24.9$_{(4.8)}$ 	& 29.5$_{(6.1)}$ \\
Prop2(H,KL)	& 12.7$_{(3.3)}$	& 14.7$_{(3.7)}$	&16.8$_{(3.8)}$		& 16.9$_{(3.8)}$	& 19.5$_{(3.7)}$ 	& 22.7$_{(4.1)}$ \\  \hline
 		   \end{tabular}
}
\end{center} 
\end{table}

We then consider the case of equal spherical covariance matrices  $\M_k=\I$ for $k=1,\ldots,\kdim$. In this case, one expects that KL-distance is better choise over ellipticity distance. 
 This set-up favors LDA over QDA due to equality of covariance matrices.   The true symmetry center $\bom \mu_1$ of the first class was the origin and for the remaining classes the mean vector $\bom \mu_k$ were taken to have norm equal to $\delta$ in  orthogonal directions. 
For Gaussian class distributions, we set $\delta=3$ and for heavy-tailed $t_2$-distributions, we set $\delta=4$. 
Table~\ref{tab:caseA} gives the estimated misclassification risk for both class distributions.  
When the class distributions are standard normal distributions,  all RDA approaches provide uniformly lower misclassification errors than LDA/QDA, but now the differences between all RDA approaches are insignificant so it is not possible to declare a winner.  In general, one can say that all RDA approaches are performing equally well. 
 In $t_2$-case, the numbers  illustrate that robust RDA approaches that are based on KL-distance provide consistently significantly  better misclassification risks (about 2--5\% improvements) to Prop1(G,KL).  
Prop2(H,E)   is not offering better performance than Prop1(G,KL) despite the robustness of the used loss function. 
 This again illustrates  the importance of choosing the right penalty: for equal class covariance matrices (and heavy-tailed distributions), $\dI$ penalty  seems  more appropriate choice than $\dE$.  This observation is also supported by comparing the performance of Prop2(G,E) to  Prop1(G,KL) which both are based on Gaussian loss function, but different distance function.  Among robust RDA approaches, Prop1(H,I) and Prop2(H,I) are performing the best. It should be noted that 
for $\pdim=10$, they offer  Oracle performance as their error rates are close to Oracle2 rule which uses the true covariance matrices.




\subsection{Data example} \label{sec:data}  

For illustrative purposes, we enclose the paper with a simple example of applying RDA on thw well-known Fisher's IRIS data \cite{fisher1936use} 
which has $K=3$  samples, each having $n_k= 50$  $p=4$-variate observations. 
We partition the original ($3 \times 50$) dataset into a training ($3 \times T$) and a validation ($3 \times V$) subsets ($T+ V=50$). 
The different T/V paritionings used were 30/20, 25/25, 15/35 and 10/40. 
To demonstrate the robustness of our techniques over the standard Gaussian tools, we replaced two measurements in each training group by outliers with relatively high random amplitudes generated as $\zeta (1,  1, 1, 1)^\top$, where $\zeta$ was generated from  $Unif(0,1024)$ for each random T/V splits of 
the datasets. The training data set is used to estimate the regularized class covariance matrices and forming the  RDA rule using
5-fold  CV procedure for penalty parameter selection. 
We then calculated the average misclassification errors  on  the validation subset and 
the results, collected in Table~\ref{tab:irisdata}, 
are averaged over $100$ random T/V partitions of the original dataset.  Also results using LDA and QDA rules are reported. 
 These figures clearly illustrate that robust RDA rules outperform the conventional LDA and QDA rules as well as  Friedman's RDA rule,  Prop1(G,KL). 
 Furthermore, note that RDA rules based on Proposition 2 are giving slightly better results compared with RDA rules based on Proposition~1. This is most evident in the case of 10/40 partitioning, which is also the case in which regularization approaches are most useful due to relatively small sample size  ($n_k=10)$. For 10/40 partitioning case,  Prop2(H,KL) gives 4.2\%  error rate whereas Prop1(H,KL)  attains 6.3\% error rate. In constrast, the conventional non-robust LDA and Friedman's Prop1(G,KL) yield 11.5\% and 9.3\% error rates, 
respectively.

\begin{table}[!t]
\vspace{-0.5cm}
\caption{ Average (\%)  validation misclassification errors for IRIS data. 
Here $30/20,\; 25/25,\; 15/35$ and $10/40$ refer to T/V random splits of  $n_k=50$ measurements in each 
 class into the training and validation subsets. The training data was used for group covariance estimation and forming the RDA rules using $5$-fold  cross validation.  
Results are averages of 100 random T/V splits of the data sets.}\label{tab:irisdata}
\begin{center}
{\small
\begin{tabular}{|  l | c |  c | c | c |c |  c | c | c |}
\hline
& 30/20 & 25/25 & 15/35 & 10/40 \\ \hline 
LDA           	& 7.0		& 6.8  	& 9.6 	& 11.5\\
QDA           	& 5.0		& 4.7		&6.3 		& 8.3 \\
Prop1(G,KL)	& 5.1      	& 4.9  	& 7.1  	& 9.3  \\
Prop1(T,E)	& 2.7      	& 3.6  	& 3.9		& 4.0  \\
Prop1(H,E)	& 2.9      	& 3.1  	& 3.9  	& 6.3  \\
Prop1(H,KL)	& 2.8      	& 3.3  	& 3.9  	& 6.4  \\
Prop2(T,E)  	& 2.8      	& 3.5  	& 3.7  	& 5.8  \\
Prop2(H,E)	& 2.8      	& 3.1  	& 3.7  	& 4.7  \\
Prop2(H,KL)	& 2.9      	& 3.4  	& 3.7  	& 5.8  \\ \hline
  \end{tabular}
}
\end{center} 
\end{table}

\section{Conclusions} \label{sec:concl} 

In this paper, we have formulated a joint penalized ML (or $M$) estimation approach for estimating the unknown scatter matrices of $\kdim>1$ samples and a joint center. 
The penalty function is based on a distance that enforces similarity. We considered three different jointly $g$-convex penalties, namely Riemannian, Ellipticity,  and KL-distance
in our formulations.

We illustrated the usefullness of our estimators in RDA setting.  In this connection,  we would like to stress that discriminant analysis is only one application where 
the developed approach can be used. We expect that our approach and framework can find uses  in many other applications such as 
radar signal processing or graphical models, where similar ideas has been used; See \cite{besson2008covariance,danaher2014joint}, for example. 
There are still room for improvements in the RDA approach. For example, we did not explore using different loss functions for different classes or using different penalties for different classes. Also, the distance (penalty) function can be different for each class. Such choices can be useful in some applications.

We did not use Frobenius  distance which is based on classical Euclidean geometry where as our approach is based on $g$-convexity which treats $\PDH(\pdim)$ as a differentiable Riemannian manifold with geodesic path \eqref{eq:geopath}. Let us point out that there are other distance functions $d(\A,\B)$ that coud be used such as 
S-divergence \cite{sra2011positive}: 
\[
d_{{\rm S}}(\A,\B) =  \log \Big| \frac{\A + \B}{2} \Big| -  \frac 1 2 \log | \A \B |. 
\]
$S$-divergence obviously satistifies D1 and it was shown in \cite{sra2011positive} that $d_{{\rm S}}$ is jointly $g$-convex, i.e., verifies D2.   
Moreover, S-divergence possesses properties similar to that of geodesic distance $d_{{\rm R}}(\cdot,\cdot)$, such as symmetry property D3  (and also affine invariance D4), but has the benefit of being easier to compute. Indeed the induced mean \eqref{mean_def} is a solution to a fixed point equation
\[
\Q  = \left( \sum_{k=1}^\kdim \pi_k \bigg(\frac{ \Q + \Q_k}{2} \bigg)^{-1}   \right)^{-1}
\]
and thus can be interpreted as weighted harmonic mean of pairwise averages. Despite of the above representation for the mean for fixed $\M_1,\ldots,\M_\kdim$, joint estimation of the scatter matrices result into rather complex estimating equations.  Therefore we omitted the use of  this distance function in our framework.

\appendix
\section{Proofs}
\noindent {\bf Proof of Proposition~1}  Properties D1, D4 and D5 are obvious. For D2, we show in Lemma \ref{tr_conv} below that $\log \Tr(\A^{-1}\B)$ is jointly $g$-convex. Next, we note that $\log | \A^{-1} \B | = \log |\A^{-1}|  + \log | \B|$, and that the log-determinant function is a $g$-linear function, i.e.\ $\pm \log | \A|$ is $g$-convex. Hence 
$\log |\A^{-1} \B|$ is jointly $g$-convex, and so D2 holds.  Since \eqref{mean_def} is a sum of $g$-convex functions, the necessary and sufficient condition for  $\M$ to be 
the solution to \eqref{mean_def} is the vanishing of the gradient, $\Nabla_{\M} \sum_{i=1}^\kdim \pi_k \, \dE(\Q_k,\Q) = \bo 0$, the solution of which is easily found to 
be \eqref{eq:Emean}. 

\begin{lemma}
\label{tr_conv}
$\log\Tr(\A^{-1}\B)$ is a jointly strictly $g$-convex function.
\end{lemma}
\begin{proof}
The geodesic curves connecting $\A_0$ with $\A_1$ and $\B_0$ with $\B_1$ on the Riemannian PSD manifold are given by:
\begin{eqnarray*}
\A_t&=&\A_0^{\frac{1}{2}}\(\A_0^{-\frac{1}{2}}\A_1\A_0^{-\frac{1}{2}}\)^t\A_0^{\frac{1}{2}}
=\A_0^{\frac{1}{2}}\U_A\D_A^t\U_A^T\A_0^{\frac{1}{2}},\\ 
\B_r&=&\B_0^{\frac{1}{2}}\(\B_0^{-\frac{1}{2}}\B_1\A_0^{-\frac{1}{2}}\)^r\B_0^{\frac{1}{2}}
=\B_0^{\frac{1}{2}}\U_B\D_B^r\U_B^T\B_0^{\frac{1}{2}},
\end{eqnarray*}
where the right hand sides are obtain from using the eigenvalue decompositions 
\[ \A_0^{-\frac{1}{2}}\A_1\A_0^{-\frac{1}{2}} = \U_A\D_A\U_A^T, \quad \mbox{and} \quad
\B_0^{-\frac{1}{2}}\B_1\B_0^{-\frac{1}{2}} = \U_B\D_B\U_B^T. \]
This gives
\begin{align*}
\log\Tr{\A_t^{-1}\B_r}
&= \log\Tr{    
 \A_0^{-\frac{1}{2}}\U_A\D_A^{-t}\U_A^T\A_0^{-\frac{1}{2}}
 \B_0^{\frac{1}{2}}\U_B\D_B^r\U_B^T\B_0^{\frac{1}{2}}}\\
&=\log\Tr{    
 \A_0^{-\frac{1}{2}}\U_A\D_A^{-\frac{t}{2}}\D_A^{-\frac{t}{2}}\U_A^T\A_0^{-\frac{1}{2}}
 \B_0^{\frac{1}{2}}\U_B\D_B^{\frac{r}{2}}\D_B^{\frac{r}{2}}\U_B^T\B_0^{\frac{1}{2}}}\\
&=\log\Tr{    
 \D_A^{-\frac{t}{2}}\U_A^T\A_0^{-\frac{1}{2}}
 \B_0^{\frac{1}{2}}\U_B\D_B^{\frac{r}{2}}\cdot\D_B^{\frac{r}{2}}\U_B^T\B_0^{\frac{1}{2}}\A_0^{-\frac{1}{2}}\U_A\D_A^{-\frac{t}{2}}}\\
 &=\log\Tr{\C\C^T},  ~ \mbox{where} ~ \C=\D_A^{-\frac{t}{2}}\U_A^T\A_0^{-\frac{1}{2}}
 \B_0^{\frac{1}{2}}\U_B\D_B^{\frac{r}{2}} \\
 &=\log\sum_{i,j}\C_{i,j}^2 = \log\sum_{i,j}\(\U_A^T\A_0^{-\frac{1}{2}}
 \B_0^{\frac{1}{2}}\U_B\)_{i,j}^2\(\D_A\)_{ii}^{-t}\(\D_B\)_{jj}^{r}\\
 &=\log\sum_{i,j}\(\U_A^T\A_0^{-\frac{1}{2}}
 \B_0^{\frac{1}{2}}\U_B\)_{i,j}^2e^{-t\log\(\(\D_A\)_{ii}\)+r\log\(\(\D_B\)_{jj}\)}.
\end{align*}
Since the log-sum-exp expression is strictly convex in $(t,r)$, the log-trace function is jointly stricly $g$-convex.
\end{proof}

\medskip

\noindent {\bf Proof of Proposition~2}  KL-divergence satisfies D1 and D4. By Lemma \ref{tr_conv}, $\Tr(\A^{-1}\B)$ is  jointly strictly $g$-convex, which implies $\dI(\A,\B)$ is jointly strictly $g$-convex, i.e., D2 holds.  Since \eqref{mean_def} is a sum of strictly $g$-convex functions, the unique minimizer is found by solving  
$\Nabla_{\M} \sum_{i=1}^\kdim   \pi_k \,  \dI(\Q_k,\Q) = \bo 0$, which gives \eqref{eq:Imean}.

\section*{References}


\end{document}